\documentclass[letterpaper]{article}                     
\usepackage{amsmath,amsthm, graphics}

\usepackage{lipsum}
\usepackage{graphicx}
\usepackage{alltt}
\usepackage{epsfig}
\usepackage{epstopdf}
\usepackage{algorithm}
\usepackage[noend]{algorithmic}
\newtheorem{fact}{Fact}
\newtheorem{theorem}{Theorem}[section]
\newtheorem{lemma}[theorem]{Lemma}

\newenvironment{Proof}[1]{\noindent{\bf Proof#1.}~}{\hfill$\qed$}
\newenvironment{relemma}[1]{\medskip\noindent{\bf Lemma #1 }\em}{\medskip}
\newcommand{\cA}{{\cal A}}
\newcommand{\cB}{{\cal B}}
\newcommand{\cD}{{\cal D}}
\newcommand{\cF}{{\cal F}}
\newcommand{\cI}{{\cal I}}
\newcommand{\cP}{{\cal P}}
\newcommand{\cS}{{\cal S}}
\newcommand{\cT}{{\cal T}}
\newcommand{\cU}{{\cal U}}
\begin{document}

\title{Broadcasting in Networks of Unknown Topology in the Presence of Swamping
\thanks{A preliminary version of this paper appeared in Proc. 12th International Symposium on Stabilization, Safety, and Security of Distributed Systems (SSS 2010), New York City, USA, September 20-22, Vol 6366, pp. 267-281, LNCS Springer, 2010.}
\thanks{Full proofs for this paper published 2010: \newline\tt{https://scs.carleton.ca/content/communication-networks-spatially-correlated-faults}}}
\author{Evangelos Kranakis
\thanks{School of Computer Science, Carleton University, Ottawa, Ontario, K1S~5B6, Canada.\tt{kranakis@scs.carleton.ca}
}
\and
Michel Paquette (corresponding author)
\thanks{Department of Computer Science, Vanier College, Montreal, Quebec, H4L~3X9, Canada. \tt{michel.paquette@vaniercollege.qc.ca}, Tel:514-744-7500 ext.7620, Fax:514-744-7505}
}

\date{}
\maketitle

\begin{abstract}
In this paper, we address the problem of broadcasting in a wireless network under a novel communication model:
the {\em swamping} communication model.
In this model, nodes communicate only with those nodes at geometric distance greater than $s$ and at most $r$ from them.
Communication between nearby nodes under this model can be very time consuming, as the length of the path between two nodes within distance $s$ is only bounded above by the diameter $D$, in many cases.
For the $n$-node lattice networks, we present algorithms of optimal time complexity, respectively $O(n/r + r/(r-s))$ for the lattice line and $O(\sqrt{n}/r + r/(r-s))$ for the two-dimensional lattice.
We also consider networks of unknown topology of diameter $D$ and of a parameter $g$ ({\em granularity}).
More specifically, we consider networks with $\gamma$ the minimum distance between any two nodes and $g = 1/\gamma$.
We present broadcast algorithms for networks of nodes placed on the line and on the plane with respective time complexities $O(D/l + g^2)$ and $O(Dg/l + g^4)$, where $l \in \Theta(\max\{(1-s),\gamma\})$.
\end{abstract}

\section{Introduction}

One of the known problems commonly faced by radio transceivers is that of {\em swamping} (cf., e.g., \cite{Bergent2002,ICGC,Radcon}).
When two wireless nodes are at close proximity, their receivers cannot adapt to strong incoming signals;
communication becomes difficult, even impossible.
In contrast to traditional radio communication models, nodes at close proximity are not able to communicate directly;
intermediate nodes are needed to relay their messages.
In this paper, we consider a wireless network where nodes suffer from the problem of swamping:
each node cannot receive any message from nodes within distance $s$ of it (the swamping distance) and may correctly receive messages only if no node within distance $s$ from them transmits.

We study analytically the problem of broadcasting in networks where nodes may be suffering from swamping.
We propose broadcasting algorithms for this novel communication model which successfully broadcast in networks of unknown topology.
Moreover, we propose algorithms to broadcast in optimal time complexity in the lattice line and in the two-dimensional lattice.

\subsection{The Model and Problem Definition}
\label{s:ModelProblem}
Typical wireless receivers are built from a radio-frequency amplifier, a demodulator and a decoder.
The amplifier adapts the strength of the received signal such that it becomes usable for the demodulator stage.
However, this amplifier is not ideal.

When the received signal strength is too low its output is either too weak or too noisy to be usable;
the first situation occurs when the communication range of a receiver is exceeded, for instance.
When the received signal strength is too high, its input stage becomes saturated leading to a distorted signal (cf., e.g., \cite{SedraSmith1998});
in this case, we say that the receiver is {\em swamped} (cf., e.g., \cite{Bergent2002,ICGC,Radcon}).
This occurs when there is a radio transmitter which is too close to a receiver.
We now propose our model for this fault phenomenon.
In what follows, whenever we speak of the distance, it is meant in its geometric sense, unless otherwise mentioned.

We work in the swamping communication model.
Our graphs are built from a set $V$ of $|V| = n$ nodes, placed on the line (Sections~\ref{s:LL} and \ref{s:Highway}) or on the plane (Sections~\ref{s:Lattice} and \ref{s:City}).
Nodes are equipped with communication range $r$ and limited by a minimum distance requirement of $s$ (the swamping distance).
Two nodes $u,v \in V$ located at distance $dist(u,v)$ greater than $s$ and at most $r$ from one-another are neighbors and share an undirected link $(u,v) \in E$ in the graph $G$;
no other links exist in $G$.
In each round, each node is either a sender or a receiver.
A node $u$ which is a transmitter in a given round sends a message to the entire set of its neighbors $\Gamma(u)$ within the same round;
this transmission also makes the receiving of messages impossible for all nodes within distance $s$.
More formally, for each round when a node within distance $s$ of it transmits, a node $v$ receives no message;
in this case, only noise is heard by $v$, indistinguishable from the background noise heard when no messages are sent.
In a fixed round, a node $v$ receives a message if and only if it is a receiver, exactly one of its neighbors is a sender, and no node within distance $s$ sends a message.
If no neighbor of $v$ is a sender, then there is no message on the channel which $v$ can receive.
If more than one neighbor of $v$ sends a message, we say that a {\em collision} occurs at $v$ and $v$ can only perceive noise on the channel.
Nodes do not have collision detection abilities, i.e., they cannot distinguish collision noise from background noise (which is apparent when no messages are heard).

The swamping communication model can be viewed as a GRN on which radio
communication is implemented with additional transient reception faults on all nodes at close proximity of a transmitter, i.e., a node cannot receive messages at each round when some node within distance $s$ of it transmits.
Alternately, we can say that all {\em incoming} links of nodes at close proximity to a transmitter fail.
Observe that communication between nearby nodes under this model can be very time comsuming as the length of the path between two nodes within distance $s$ is only bounded above by the diameter $D$, in many cases.

Throughout this paper, we study networks of nodes placed on the line and on the plane which are either designed (sections~\ref{s:LL} and \ref{s:Lattice}) or of {\em unknown topology} (sections~\ref{s:Highway} and \ref{s:City}).
Nodes are {\em location-aware},
i.e., each node knows its own location with respect to some global reference, but all nodes are unaware of the location of any other node.
In the cases where the topology is unknown, we restrict attention to connected networks where nodes are positioned with some minimum distance $\gamma$ from each other.
The parameter $\gamma$ may be related to the physical size of the nodes such that no two could occupy the same space.
Let the parameter $g=1/\gamma$ be called {\em granularity} (as introduced in \cite{2007EGKPPS}).
Nodes are also aware of the parameter $\gamma$, the {\em swamping distance} $s$, and the communication distance $1$.

We consider the process of broadcasting under the {\em spontaneous wake up model} in which all nodes are considered to be awake when the source begins transmission.
Under this model, nodes may contribute to the broadcasting process even before receiving the source message, by exchanging control messages.
In the sequel, we consider that nodes execute algorithms in a synchronous way.

We consider deterministic algorithms without global knowledge (Sections~\ref{s:LL}, \ref{s:Highway}, and \ref{s:Lattice}) and with some knowledge about messages received by nodes close by (Section~\ref{s:City}).
In general, the algorithm is known to all nodes and its execution is based solely on the location of nodes in the network, the history known to each node, and the parameter $g$.
In Section~\ref{s:City}, the algorithm execution is based on the above-mentioned information augmented by the information about messages received by nodes surrounding each node.

\subsection{Our Results}
In Section \ref{s:LL}, we address the problem of broadcasting on the lattice line and show a broadcasting algorithm, $\cA$, which correctly broadcasts the message $m$ on the lattice line of length $n$, in time $\lfloor n/r \rfloor + 3(\lceil r / (r-(s+1)) \rceil + 1)$.
This order of magnitude for the time complexity is optimal.

In Section~\ref{s:Highway}, we provide an algorithm, $\cB$, to correctly broadcasts a message $m$ in a network of unknown topology in the line.
Given a network diameter $D$, a minimum distance between nodes $\gamma$, a granularity parameter $g = 1/\gamma$ and $l = \max\{(1-s),\gamma\}$, Algorithm $\cB$ completes broadcasting in time $O(D/l + g^2)$.

In Section \ref{s:Lattice}, we address the problem of broadcasting on the two-dimensional lattice and show a broadcasting algorithm, $\cA^2$, which correctly broadcasts the message $m$ in the two-dimensional lattice line of length $n$, in time $4\lfloor \sqrt{n}/r \rfloor + 12(\lceil r / (r-(s+1)) \rceil + 1)$.
This order of magnitude for the time complexity is optimal.

In Section~\ref{s:City}, we provide an algorithm, $\cB^2$, to correctly broadcasts a message $m$ in a network of unknown topology in the plane.
Given a network diameter $D$, a minimum distance between nodes $\gamma$, a granularity parameter $g = 1/\gamma$ and $l = \max\{(1-s)/(3\sqrt{2}),\gamma/\sqrt{2}\}$, Algorithm $\cB^2$ completes broadcasting in time $O(Dg/l + g^4)$.

\section{Related Work}
\label{s:RelatedWork}

The fundamental questions of network reliability have received much attention in the context of wired networks, under the assumption that components fail randomly and independently (cf., e.g. \cite{Bienstock1988,CDP1994,CDP1996,PP2006_2} and the survey \cite{Pelc1996}).
On the other hand, empirical work has shown that positive correlation of faults is a more reasonable assumption for networks \cite{GGSE2001,TJ2001,YKT1996}.
In particular, in \cite{YKT1996}, the authors provide empirical evidence that data packets losses are spatially correlated in networks.
Moreover, in \cite{GGSE2001}, the authors state that the environment provides many phenomena that may lead to spatially correlated faults.
More recently, in \cite{KPP2007}, a gap was demonstrated between the fault-tolerance of networks when faults occur independently as opposed to when they occur with positive correlation.
To the best of our knowledge, this was the first paper to provide analytic results concerning network fault-tolerant communication in the presence of positively correlated faults for arbitrary networks.

In contrast, few results are known about fault-tolerant communication in geometric radio networks.
In \cite{KKP2001}, the authors consider the problem of broadcasting in a fault-free connected component of a radio network whose nodes are located at grid points of square grids and can communicate within a square of size $r$.
For an upper bound $t$ on the number of faulty nodes, in worst-case location, the authors propose a $\Theta(D + t)$-time oblivious broadcast algorithm and a $\Theta(D + \log(\min(r,t)))$-time adaptive broadcast algorithm, both operating on a connected fault-free component of diameter $D$.
More recently, the authors of \cite{CMPS2007} present a different problem, that of gossiping in directed GRNs with transient faults.
When nodes may send a single message per time slot, the authors present an algorithm performing in time $O(n\Delta)$ with $O(n^2)$ messages and show these bounds to be optimal.
When nodes may send multiple messages per time slot, they provide an algorithm functioning in optimal time complexity $O(D\Delta)$ and message complexity $O(Dn)$.
The same algorithm performs broadcasting within optimal time and optimal message complexity $O(n)$.
Also, in \cite{KPP2008}, an algorithm was demonstrated to broadcast correctly with probability $1-\epsilon$ in faulty random geometric radio networks of diameter $D$, in time $O(D + \log 1/\epsilon)$.
In \cite{CMPS2012}, the work from \cite{CMPS2007} is extended with the presentation of a distributed algorithm capable of broadcasting in all connected GRNs of unknown topology, with time complexity $O(DR^2/\gamma^2)$, where $R$ is the maxiumum range of transmitters and $\gamma$ is the minimum distance between receivers.
If we let $R=1$ (for simplicity) and $g = 1/\gamma$, as in our paper, this translates to a time complexity of  $O(Dg^2)$.
We provide an algorithm that broadcasts in time $O(Dg/l + g^4)$, where $l \in \Theta(\max\{(1-s),\gamma\})$.
In comparison, this is always at least as fast when $D \in \Omega(1/\gamma^2)$, i.e., on large and/or dense networks, and faster in cases where $\gamma \in o(1-s)$, that is when the degree of nodes is possibly very large.

The question of communication in networks of unknown topology has been widely studied in recent years.
In \cite{CMS2004}, the authors state that broadcasting algorithms which function in unknown GRNs also function in the resulting fault-free connected components of faulty GRNs.
A basic performance criterion of broadcasting algorithms is the time necessary for the algorithm to terminate;
in synchronous networks, this time is measured as the number of communication rounds.
For networks whose fault-free part has a diameter $D$, $\Omega(D)$ is a trivial lower bound on broadcast time, but optimal running time is a function of the information available to the algorithms (cf., e.g., \cite{DP2007}).
For instance, in \cite{DP2007}, an algorithm was obtained which accomplishes broadcast in arbitrary GRNs in time $O(D)$ under the assumption that nodes have a large amount of knowledge about the network,
i.e. given that all nodes have a {\em knowledge radius} larger than $R$, the largest communication radius.
The authors also show that algorithms broadcasting in time $O(D + \log n)$ are asymptotically optimal, for unknown GRNs when nodes communicate spontaneously and either can detect collisions or have knowledge of node locations at some positive distance $\epsilon$, arbitrarily small.

More recently, in \cite{EGKPPS2009}, it was shown that the time of broadcast depends on the network diameter $D$ and the smallest geometric distance $\gamma$ (denoted $d$ in their paper) between any two nodes.
Under the conditional wake-up model\index{network!conditional wake up model}, where nodes start transmitting only after hearing a first message, the authors proposed an algorithm that completes broadcasting in time $O(Dg)$.
They also proved that, in this context, every broadcasting algorithm requires $\Omega(D\sqrt{g})$ time.
Under the spontaneous wake up model\index{network!spontaneous wake up model}, where nodes may transmit from the beginning of the communication process, the authors combined two sub-optimal algorithms into one algorithm, which completes broadcasting in optimal time $O(\min(D + g^2, D \log g)$.
The results in \cite{EGKPPS2009} hold under the assumption that nodes can communicate with other neraby nodes.
We, on the other hand, consider the communication model where nodes are prevented from communicating with other nodes nearby.

In \cite{EKP2008}, under the conditional wakeup model, $\Omega(Dg)$ was shown to be the tight lower bound on broadcasting time.
However, for networks where nodes locations are restricted to the vertices of a grid of squares of size $\gamma$, the authors proposed an $O(Dg^{5/6}\log g)$-time broadcasting algorithm,
thus showing that the broadcast time is not always linearly dependent on $g$.

In \cite{FP2008}, the problem of broadcasting in unknown topology networks was proposed given that nodes do not perceive their location accurately and that they do not know the minimum distance $\gamma$ between them.
Under the spontaneous wake up model, the authors showed a broadcasting algorithm maintaining optimal time complexity $O(\min(D + g^2, D \log g)$ in these conditions given an upper bound $\gamma/2$ on the inaccuracy of node location perception;
beyond this upper bound on inaccuracy, the authors showed that broadcasting is impossible.
The solution proposed in \cite{FP2008} uses the election of {\em ambassadors} that represent a large number of nodes and communicate information to regions of the graph in range.
In contrast, we show the impossibility of using this mechanism in the presence of swamping.

In 2003, Kuhn, Wattenhofer and Zollinger \cite{KuWaZo03}, introduced a variant of the UDG model handling transmissions and interference separately, named {\em Quasi Unit Disk Graph (Q-UDG)} model.
In this model, two concentric discs are associated with each station, the smaller representing its communication range and the larger representing its interference range.
In our work, we consider a very different situation:
as in traditional radio communication models, interference and communication ranges are equal;
contrary to previous work, we add the swamping range - a self-interference range - which must be smaller than the communication range.


In the present paper, we assume that nodes communicate spontaneously, but know nothing of the network, other than their own location, and cannot detect collisions.
We propose algorithms to broadcast in networks embedded in the line and in the plane under the swamping communication model.
Contrary to the traditional radio communication model, it is not possible for nodes under the swamping communication model to directly receive messages from nodes located at close proximity to them.

\section{Lattice Line}
\label{s:LL}
Throughout this section, we assume that $r$ and $s$ are positive integers.
Consider a set of $n$ nodes placed at points $0, 1, \ldots, n-1$ in one-dimensional Euclidean space;
nodes are labeled according to their location.
We call this placement of the nodes the {\em lattice line}.
For simplicity, the communication and swamping ranges $r$ and $s$ are integer values;
each node may reach nodes which are located on points at distance at least $s+1$ from it, and at most $r$ from it.
In this section, we consider the broadcasting of a message $m$ from the node $0$ to all other nodes of the line.

In the sequel, we will present an algorithm $\cA$ and then prove the following result:

\begin{theorem}
\label{t:algoA}
Algorithm $\cA$ correctly broadcasts the message $m$ on the lattice line of length $n$, in time $\lfloor n/r \rfloor + 3(\lceil r / (r-(s+1)) \rceil + 1)$.
This order of magnitude is optimal.
\end{theorem}

\subsection{Non-Connectivity}
Consider the case when $s > 0$ and $r - s = 1$.
In this case, completing the broadcasting process in the lattice line is impossible.

\begin{lemma}
\label{l:impossible}
If $s > 0$ and $r - s = 1$, then broadcast is impossible on the lattice line.
\end{lemma}
\begin{proof}
The source node has label $0$ and possesses the source message $m$.
Since $r-s = 1$, each node $u$ has a link to a node $v$ only if it is exactly at distance $r$ from it.
We may thus model any path of this network by a sequence of additions and subtractions of $r$ on the source node label.
Hence, the node $0$ has paths only to nodes whose labels are multiples of $r$.
For $r$ ans $s>0$ integers, we have $r>1$.
In this case, the network is disconnected; broadcasting is impossible.
\end{proof}

\subsection{Fast Broadcast}
In the previous section, we have shown conditions under which broadcast is impossible.
We now show that, when these conditions are not met, broadcast is possible.
We further show an algorithm for broadcasting in optimal time on the line.

Consider two sets of nodes on the line, $A_k$ and $B_k$, where $A_k$ (respectively, $B_k$) is the set of all nodes whose labels are in the interval $[k, r - 1 + k]$ (resp., $[r + k, 2r - 1 + k]$).
We now describe the communication scheme $Local_k$ for disseminating a message inside these sets.
The scheme $Local_k$ consists of $x = \lceil r / (r-(s+1)) \rceil$ steps $i = 0, 1, \ldots, x-1$, each taking two rounds.
For each step $i$, in the first round the node with label $a_{i,k} = i \cdot (r - (s+1)) + k$ from the set $A_k$ transmits the message $m$;
in the second round the node in set $B_k$ with label $b_{i,k} = a_{i,k} + r$ transmits the message $m$.
In the following lemma, we assume the absence of collision with nodes external to the communication scheme.

\begin{lemma}
\label{cl:possible}
Given that the node $k$ has previously received the message $m$, all nodes in $A_k$ will have received the message $m$ at the end of scheme $Local_k$.
\end{lemma}
\begin{proof}
Without loss of generality, let $k = 0$.
We show by induction that the scheme $Local_0$ sends the message $m$ to all nodes in the set $A_0$.

\noindent {\bf Base step:}\newline
In step $0$, first node $a_{0,0} = 0$ transmits, the message $m$;
the nodes with labels $s+1, s+2, \ldots, r$ receive the message.
In the second round of step $0$, the node $b_{0,0} = r$ transmits;
the nodes with labels $0, 1, \ldots, r - (s+1)$ receive the message.

\noindent {\bf Inductive hypothesis:}\newline
Assume that the nodes with labels in $[0, i \cdot (r - (s+1))]$ have received the message by the end of of step $i-1$.
Then, by the end of step $i$, the nodes with labels in $[0, (i+1) \cdot (r - (s+1))]$ will have received the message.

\noindent {\bf Proof of the inductive hypothesis:}\newline
In the first round of step $i$, the node $a_{i,0} = i(r-(s+1))$ sends the message and the nodes with labels in
$$
[i \cdot (r - (s+1)) + (s+1), i \cdot (r - (s+1)) + r]
$$
receive $m$ from the node $a_{i,0} = i \cdot (r - (s+1))$.
In the second round of step $i$, the node $b_{i,0} = i(r-(s+1)) + r$ sends the mesage and the nodes with labels in
$$
[i \cdot (r - (s+1)), i \cdot (r - (s+1)) + r - (s+1)]
$$
receive $m$ from the node $b_{i,0} = r + i \cdot (r - (s+1))$.
The latter interval overlaps the set of nodes which had previously received the message $m$.
The largest label of the nodes in the receiving interval may be rewritten as $(i+1) \cdot (r - (s+1))$.
This proves the inductive hypothesis.

Hence, after $\lceil r / (r-(s+1)) \rceil$ steps of scheme $Local_k^*$, the nodes with labels in the interval $[0,(\lceil r / (r-(s+1)) \rceil) \cdot (r - (s+1))]$ will have received the message.
Since $\lceil r / (r-(s+1)) \rceil (r - (s+1)) \geq r$, this proves the lemma.
\end{proof}

Consider now the scheme $Local_k^*$ consisting of the scheme $Local_k$ to which one step is added: the step $\lceil r / (r-(s+1)) \rceil + 1$.

\begin{lemma}
\label{r:goodAB}
Given that the node $k$ has previously received the message $m$, all nodes in $A_k$ and $B_k$ will have received the message $m$ at the end of scheme $Local_k^*$.
\end{lemma}
\begin{proof}
To obtain the set of labels of the nodes in $B_k$ we need only add $r$ to each label contained in the set $A_k$.
Hence, to prove the lemma, we need only show that if the nodes with labels in $[a,b] \cap A_k$ are informed at step $i$, then, at step $i+1$, the nodes with labels in $[a+r,b+r] \cap B_k$ are informed.
For simplicity and without loss of generality, we set $k = 0$ in the following proof.

Observe that in step $i$, the node with label $i \cdot (r - (s+1)) + r$ in interval $B_k$ transmits the message to the nodes in subinterval
$$
[i \cdot (r - (s+1)), i \cdot (r - (s+1)) + r - (s + 1)]
$$
of interval $A_k$.
Also observe than in step $i+1$, the node with label $(i+1) \cdot (r - (s+1))$ in interval $A$ transmits the message to the nodes in the subinterval
$$
[(i+1) \cdot (r - (s+1)) + (s+1), (i+1) \cdot (r - (s+1)) + r]
$$
of interval $B_k$.
To see that these intervals differ exactly by $r$, we subtract $r$ from the latter and obtain
\begin{eqnarray*}
&[(i+1) \cdot (r - (s+1)) + (s+1) - r, (i+1) \cdot (r - (s+1)) + r - r]&\\
&[i \cdot (r - (s+1)), (i+1) \cdot (r - (s+1))]&
\end{eqnarray*}
which corresponds to the former.
The lemma follows.
\end{proof}

Consider algorithm $\cA$ for broadcasting on a line which consists of $2$ parts.
In the first part, the message is sparsely transmitted throughout the line by transmitting the message $m$ sequentially by nodes $0, r, \ldots, \lfloor n/r \rfloor r$.
In the second part, the scheme $Local_k$ is executed, in parallel, for $k = 0, 4r, \ldots, 4 (\lfloor n/(4r) \rfloor -1) r$ and then for $k = 2r, 6r, \ldots, (4 (\lfloor n/(4r) \rfloor -1) - 2) r$.
The algorithm is terminated by executing scheme $Local_k$ for $k = n - 2r$.

We extend the validity of Algorithm $\cA$ to any source node by re-labeling nodes sequentially left to right such that the source has label $0$.
The sparse transmission part of the algorithm is then executed from node $0$ through the positive node labels and then through the negative node labels.
The remainder of the algorithm is identical.

We now prove the main theorem of this section.

\begin{Proof}{ of Theorem~\ref{t:algoA}}
We begin by showing the algorithm correctness and then its execution time.

From Lemma~\ref{r:goodAB}, the scheme $Local_k$ will successfully broadcast the message for each interval of size $2r$, given that the node with label $k$ has received the message, and assuming no collision caused by transmissions by nodes external to the communication intervals.
Since all transmissions that are executed in parallel originate at nodes whose labels differ by $4r$ in our algorithm, no collision occurs.
Hence, if the sparse transmission part of the algorithm is successful, then the second part of the algorithm will also be successful.
Moreover, if the second part of the algorithm is successful, the node whose label is $n - 2r$ will have received the message before the termination phase of the second part.
Hence, the termination phase of the algorithm will also be successful.
Thus the algorithm correctly broadcasts the message $m$.

We now count the number of rounds necessary to execute the algorithm.
The number of rounds to complete the first part of $\cA$ is $\lfloor n/r \rfloor$.
Each phase of the second part and the termination phase of the algorithm are completed in time $\lceil r / (r-(s+1)) \rceil + 1$.
Hence, the algorithm is completed in $\lfloor n/r \rfloor + 3(\lceil r / (r-(s+1)) \rceil + 1)$ rounds.

We will now prove optimality.
In order to disseminate the message from one end of the line to the other, $n/r$ constitutes a trivial lower bound.
For one node to receive a message, only one node within distance $r$ must send a message.
Therefore, avoiding all collisions, at most $2(r - (s+1))$ nodes learn the message within each interval of size $2r$, at each round.
It follows that, for all nodes of any interval of such size to know the message $m$, it takes at least $r/(r - (s+1))$ steps.
Hence, the lower bound on transmission time is $\Omega(n/r + r/(r - (s+1)))$.
\end{Proof}

\section{Highway Model}
\label{s:Highway}
In this section, we analyze the problem of broadcasting along a line segment of length $L$ where nodes are placed by an adversary.
Each node $u$ is equipped to communicate with all nodes that are both within distance $1$ and at distance greater than $s$ from it.
Hence, in this section, we assume that $r=1$ for simplicity.
More formally, we describe the {\em highway model}.
The communication range of a node $u$ is the interval within distances $(s, 1]$ from $u$.
The {\em size} of the communication range is the length of this interval,
i.e., $1-s$.
The adversary designs the network such that it is {\em connected} and the distance between any pair of nodes $u,v$ is at least $\gamma$.
We remind that the parameter $\gamma$ may be related to the physical size of the nodes such that no two could occupy the same space.
We say that a network is connected if, for any node pair $u,v$, there exists a path in the network from node $u$ to node $v$.
Observe that, if $\gamma> 1$, then the network is not connected.

Message collisions result in noise indistinguishable from background noise and nodes are not equipped to detect these collisions.
However, in this section we use the apparent silence from collisions to discover the presence of nodes through a {\em collision-causing} algorithm.

Nodes are aware of the parameter $\gamma$ (and $g=1/\gamma$) and the coordinate system of the line segment of length $L$.
Each node also knows the parameter $s$, its swamping distance, and its communication distance $1$.

In this section, we present a broadcasting algorithm $\cB$ and show the following result.

\begin{theorem}
\label{t:algoB}
Algorithm $\cB$ broadcasts a message $m$ in a network of diameter $D$ in time $O(D/l + g^2)$, where $l = \max\{(1-s),\gamma\}$.
\end{theorem}

In order to prove the main result of this section, we need several preparatory lemmas.
The following fact requires no proof.

\begin{fact}
\label{f:connect}
For any node $u$ in a connected network, there is at least one node $v$ within the set $\Gamma(u)$.
\end{fact}

\subsection{Partition $\cP$ of the Line}
We now define a partition, called $\cP$, on which our communication algorithm will operate.
For each line segment in the partition below, the segment includes its leftmost point and excludes its rightmost point so that there will be no intersection between adjacent segments.
We provide a graphical representation of the partition in Figure \ref{fig:partition} and describe it below in detail.
\begin{figure}[!htb]
\centering
\includegraphics[scale=0.75, bb =0 0 386 256,angle=0]{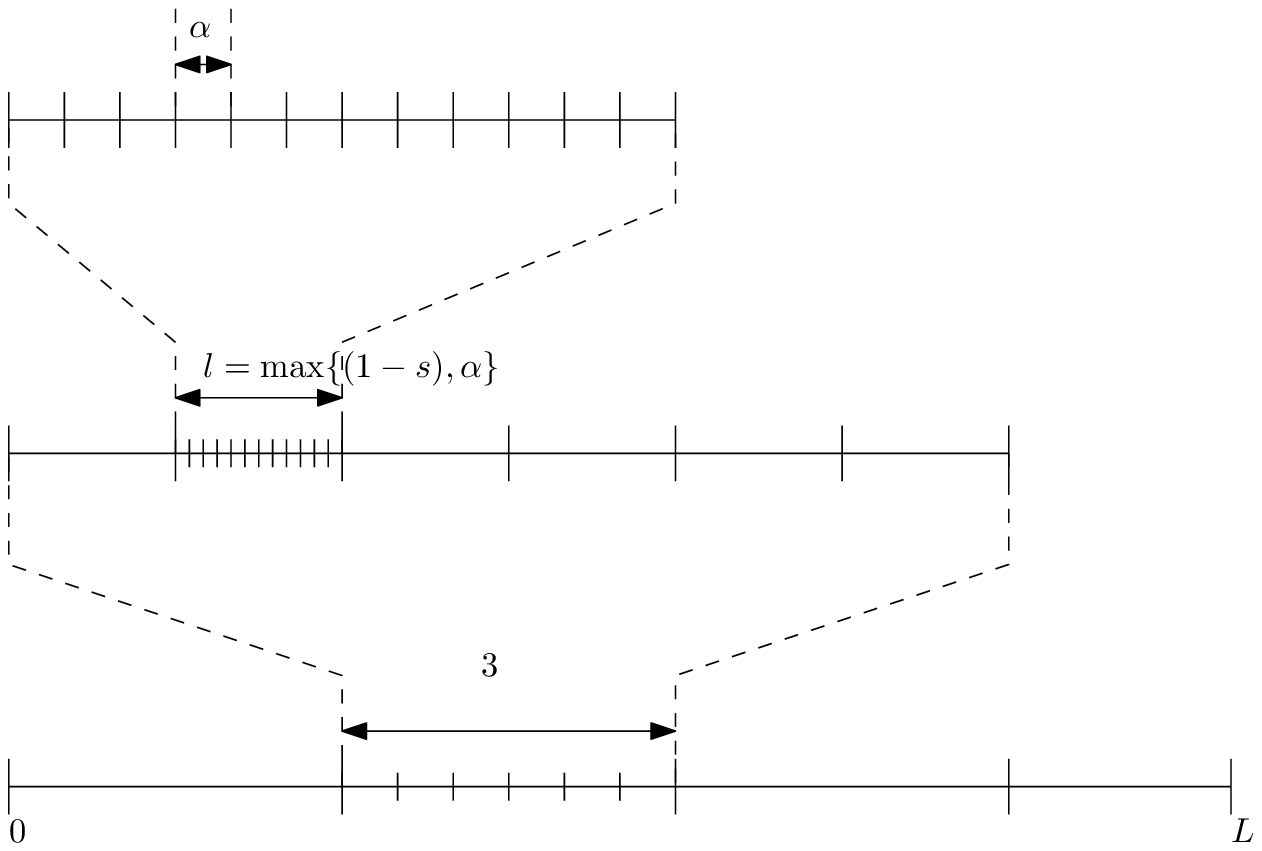}
\caption{Partition $\cP$}
\label{fig:partition}
\end{figure}

Partition the line into line segments of length $3$, called {\em regions}.
The line contains $\lceil L/3 \rceil$ regions, where $\lfloor L/3 \rfloor$ are of length $3$ and at most one (the rightmost) is shorter, and even may consist of a single point.

Further, partition each region into smaller line segments, called {\em blocks}, of length $l = \max\{(1-s),\gamma\}$.
Here, $l \leq 1$ since both $\gamma \leq 1$ and $1-s \leq 1$.
Each region contains $\mu = \lceil 3/l \rceil$ blocks, where $\lfloor 3/(1-s) \rfloor$ are of length $l$ and at most one (the rightmost) is shorter, and even may consist of a single point.
For each region, label blocks $1, 2, \ldots, \mu$, from left to right.

Partition also each block into line segments of length $\gamma$, called {\em homes}.
Each block contains $\nu = \lceil l/\gamma \rceil$ homes, where $\lfloor l/\gamma \rfloor$ are of length $\gamma$ and at most one (the rightmost) is shorter, and even may consist of a single point.
For each block, label homes $1, 2, \ldots, \nu$, from left to right.

\subsubsection{Partition Properties}
We now show communication properties related to the partition defined above.
We first show that transmissions in  distinct regions do not collide, if properly scheduled.
We then show that transmissions by a few distinguished nodes in a part or all of a block can reach all neighbors of nodes on this block or part of a block.
However, before showing these properties, we observe that since homes are of length at most $\gamma$, at most one node can occupy each home.
Hence, we have the following lemma.
\begin{lemma}
\label{l:singleHome}
Each home contains at most one node.
\end{lemma}

\begin{lemma}
\label{l:noCollide}
Transmissions from unique nodes inside identically labeled blocks in distinct regions do not collide.
\end{lemma}
\begin{proof}
Consider nodes $u,v$ in different regions and identically labeled blocks.
Each region has length $3$ and each block has length $l \leq 1$.
Because the block labels are identical in each region, the minimum distance between two identically-labeled blocks (that contain the nodes $u,v$) is $3 - l \geq 2$.
Since each line segment of the partition excludes its rightmost point, there is no point within distance $1$ of both $u$ and $v$.
\end{proof}

\begin{lemma}
\label{c:rangeOverlap}
Consider any pair $u,v$ of nodes within distance $1-s$.
Also consider the set $\cU$ of all nodes inclusively located between $u$ and $v$.
We have that $\Gamma(\cU) = \Gamma(u) \cup \Gamma(v)$.
\end{lemma}
\begin{proof}
Consider two nodes $u,v$ at distance $d \leq 1-s$ from one-another;
$u$ is to the left of $v$ and $u$ is at coordinate $0$.
Consider the right part of the range of $u$ and $v$.
Then, the range of $u$ to the right covers the interval $(s,1]$.
Similarly, the range of $v$ to the right covers the interval $(s+d, 1+d]$.
Since $0 < d \leq 1-s$, we have that $s + d \leq s + 1 - s = 1$.
Hence, the functional portions of the ranges of $u$ and $v$ overlap and cover the interval $(s, 1 + d]$.

Consider any node $w$ between $u$ and $v$, i.e., at distance $dist(u,w)$ from $u$, with $0 < dist(u,w) < d$.
The range of $w$ to the right covers the interval $(s + dist(u,w), 1 + dist(u,w)]$.
Since $0 < dist(u,w) < d$, the range of $w$ to the right is completely included in the ranges of $u$ and $v$ to the right.
The argument is symmetric for the left.
\end{proof}

In the following sections, we describe communication procedures that will enable nodes to broadcast messages to all nodes of their networks.

\subsection{Procedure $\cD^*$ for Neighborhood Discovery}
We now define procedures used to communicate once from each node to all other nodes within distance $1$ of them.
We refer to this process as {\em Neighborhood Discovery}.

\subsubsection{Procedure $\cD$}
We now present Procedure $\cD$, in which nodes in distinct homes inside a region sequentially send a message while other nodes listen.
This procedure is executed in parallel over all regions and for all $(block, home)$ labels sequentially.
All homes with some $(block, home)$ label transmit a message while all other nodes listen for incoming messages.
More formally, refer to the code for Procedure~$\cD$.
\begin{algorithm}{}
{\bf Procedure $\cD$}
\begin{algorithmic}
\STATE {\bf In parallel for all regions}
\STATE $N_u \leftarrow \emptyset$ // the set of nodes known to $u$
\STATE $H_u \leftarrow$ the (block, home) label of $u$
\FOR{$block = 1..\mu$}
    \FOR{$home = 1..\nu$}
	   \IF{$H_u = (block, home)$}
            \STATE Transmit $hello$
        \ELSIF{a $hello$ is  heard}
            \STATE $N_u \leftarrow N_u \cup (block, home)$
        \ENDIF
    \ENDFOR
\ENDFOR
\end{algorithmic}
\end{algorithm}
By Lemma~\ref{l:noCollide}, no collision occurs in this procedure.
Hence we claim that each node will gain knowledge of all nodes located within its communication range as a result of Procedure $\cD$.
We have the following lemma:
\begin{lemma}
\label{l:procD}
After one execution of Procedure $\cD$, nodes know of all nodes within distance $1$ and greater than $s$ of them.
\end{lemma}
\begin{proof}
Fix any node $u$ and the set $\Gamma(u)$ of all nodes within its communication range.
By Lemma~\ref{l:noCollide}, no message collision can occur during Procedure $\cD$.
Since the procedure makes nodes in all the $(block, home)$ couples transmit, it then follows that $u$ must receive messages from all the nodes in the set $\Gamma(u)$.
\end{proof}

By the previous lemma, since the graph is connected, each node will discover at least one node within its communication range by the end of Procedure $\cD$.
This fact allows all nodes to discover all other nodes within distance $s$ of them by Procedure $\cD^*$.

\subsubsection{Procedure $\cD^*$}
Recall that, in our model, it is not possible for any node to hear messages from nodes within distance $s$ of them.
Observe that the length of the path between two nodes within distance $s$ is only bounded above by the diameter $D$, in many cases.
We now concentrate on a time-guaranteed procedure for discovery of nodes within distance $1$.

Consider Procedure $\cD_{(b, h)}$ which uses the absence of distinguishable messages from collisions to discover nodes within distance $s$.
During this entire procedure, the node located in block $b$ and house $h$, if it exists, transmits a hello message;
all other nodes also transmit a hello message according to a schedule determined by their $(block, home)$ identifiers.
If the node at $(b,h)$ exists and is known, each turn when no message is heard reveals the presence of a node at $(block, home)$.
This detection method using collisions was proposed in a different context in \cite{KP2002}.
\begin{algorithm}{}
{\bf Procedure $\cD_{(b, h)}$}
\begin{algorithmic}
\STATE // $N_u$ is the set of nodes known to $u$
\STATE // $H_u$ is the (block, home) label of $u$
\STATE {\bf In parallel for all nodes $u \in V$}
\FOR{$block = 1..\mu$}
    \FOR{$home = 1..\nu$}
    	\IF{$H_u = (block, home)$ {\bf OR} $H_u = (b, h)$}
	       \STATE Transmit $hello$
	   \ELSIF{no $hello$ is heard {\bf AND} $(b, h) \in N_u$}
	       \STATE $N_u \leftarrow N_u \cup (block, home)$
	   \ENDIF
    \ENDFOR
\ENDFOR
\end{algorithmic}
\end{algorithm}

\begin{lemma}
\label{l:procDbh}
By Procedure $\cD_{(b,h)}$, nodes neighbor to $(b,h)$ know all other nodes within distance $1$ of them in time $\Theta(g)$.
\end{lemma}
\begin{proof}
The time complexity of Procedure $\cD_{(b,h)}$ is in $\Theta(\mu\nu)$.
With $\gamma \leq l \leq 1$, we have that $\mu = \lceil 3/l \rceil \in \Theta(1/l)$ and $\nu = \lceil l/\gamma \rceil \in \Theta(l/\gamma)$.
Hence,
$$
\mu\nu \in \Theta((1/l)(l/\gamma)) = \Theta(1/\gamma) = \Theta(g).
$$

We now prove correctness.
Fix a node $u$ which shares a link with the node $(b,h)$.
During the execution of Procedure $\cD_{(b,h)}$, the node $(b,h)$ will transmit messages at every round.
A message from $(b,h)$ will be heard by $u$ at every round when no collision occurs at $u$.
Furthermore, when no message can be distinguished, another node within distance $1$ of $u$ must be transmitting from the home with label $(block,home)$ (as defined in the procedure).
Since Procedure $\cD_{(b,h)}$ schedules all nodes to transmit in pairs with $(b,h)$,
upon completion of this procedure, the node $u$ will have discovered all nodes $w$ for which the distance $dist(u,w)$ from $u$ is at most $1$.
\end{proof}

Now consider Procedure $\cD^*$ consisting of one execution of Procedure $\cD$ followed by the execution of Procedure $\cD_{(b,h)}$ for all $(b,h) \in \{1, 2, \ldots, \mu\} \times \{1, 2, \ldots, \nu)\}$.
In plain words, Procedure $\cD^*$ schedules colliding transmissions for all $(block, home)$ couple pairs
$$
((b,h), (b', h')) \in \{\{1, 2, \ldots, \mu\} \times \{1, 2, \ldots, \nu)\}\}^2.
$$
More formally, refer to the pseudo code for Procedure $\cD^{*}$.
\begin{algorithm}{}
{\bf Procedure $\cD^{*}$}
\begin{algorithmic}
\STATE {\bf Call} Procedure $\cD$
\FOR{$b = 1..\mu$}
    \FOR{$h = 1..\nu$}
    	\STATE {\bf Call} Procedure $\cD_{(b, h)}$
    \ENDFOR
\ENDFOR
\end{algorithmic}
\end{algorithm}

\begin{lemma}
\label{l:procD*}
By Procedure $\cD^*$, nodes know all other nodes within distance $1$ of them in time $\Theta(g^2)$.
\end{lemma}
\begin{proof}
The time complexity of Procedure $\cD_{(b,h)}$ is in $\Theta(\mu\nu)$.
Hence, the time complexity of Procedure $\cD^*$ is in $\Theta(\mu^2\nu^2)$.
By the above and by Lemma~\ref{l:procDbh}, the time complexity of Procedure $\cD^*$ is therefore in $\Theta(g^2)$.

We now prove correctness.
By Fact~\ref{f:connect}, for any node $u$, since the graph is connected, there exists a node $(b,h)$ such that Lemma~\ref{l:procDbh} will hold.
By the above and by Lemma~\ref{l:procDbh}, all nodes know all other nodes that are within distance $1$ of them.
\end{proof}

It remains open whether or not the time $\theta(g^2)$ for neighbourhood discovery is optimal.

One degenerate case of swamping is when $s < \gamma$;
then, swamping has no real effect on the network, which becomes identical to a congruent GRN.
In that case, the process of discovering the neighbours takes only the time necessary for all nodes to announce their presence once, $\theta(g)$.
This is only true because of the absence of nodes with which nearby nodes can not communicate.

However, once we have $s \geq \gamma$, some links of the congruent GRN are deleted in the network with swamping;
to discover nodes at close proximity then becomes a non-trivial, collaborative task.
When messages must remain small, it is impossible to share locations of many other nodes to speed up the neighbourhood discovery process.

For small enough $s$ and unbounded message size, the task of neighbourhood discovery may be sped up.
In the 1-dimensional case, the length of paths between nearby nodes seems bounded by small enough value $\beta$.
Therefore, if nodes transmit long messages containing their current known mapping of neighbour nodes on a turn basis, repeating this process a small number of times would be sufficient to perform the neighborhood discovery process, i.e., in time $\in \Theta(\beta g)$.
We remind the reader however, that we are studying the case where messages do not have unbounded length.

With knowledge of all nodes within distance $1$, nodes have the basic tools to select distinguished nodes to relay messages for all nodes of a block.
We discuss such a procedure in the following subsection.

\subsection{Selection of Spokesman Nodes}
We now describe a procedure for selection of distinguished nodes for each block known as {\em spokesmen}.
We wish to select these spokesmen in order to avoid collisions and speed up the broadcasting process.
Before we concentrate on the different cases, we present the following fact.

\begin{fact}
\label{f:knownLocation}
Given location-awareness, if a sender includes its location inside a message, then a receiver can determine all points where the message may be received.
\end{fact}
\begin{proof}
The sender knows its own location and therefore can incorporate this as part of his message.
The receiver then knows the origin of the received message and hence can determine the covered region.
\end{proof}

Consider the spokesman selection procedure that elects, for each block,
\begin{enumerate}
\item {\em right (left) boundary spokesmen}: the node in the rightmost (leftmost) home known to be completely contained within the transmission range of a sender, if this home is the rightmost (leftmost) home on the block;
\item {\em right (left) range spokesmen}: the node in the rightmost (leftmost) home known to be completely contained in the transmission range of a sender, if this home is not the rightmost (leftmost) home on the block;
\item {\em right (left) potential spokesmen}: the node in the rightmost (leftmost) home known to be partially contained within the transmission range of a sender.
\end{enumerate}

We now show that the spokesman selection procedure making the above selections selects unique spokesmen for each type.
\begin{lemma}
\label{c:uniqueSpokesmen}
The spokesman selection procedure selects at most one node for each spokesman type.
\end{lemma}
\begin{proof}
Given that right and left boundary spokesmen are unique by definition (those nodes in the home that is closest to the block boundaries), we prove the lemma for right and left range and potential spokesmen.
In the case when $l = \gamma$, there is only one home per block, hence the lemma holds in this case.
We now prove the lemma for the case when $l = 1-s$.

Given that the functional portion of the communication range of a node is of size $1 - s = l$, the range of a transmitter always encloses at least one of the homes that is closest to the block boundaries;
call this home a boundary home.
For any set $\cS$ of transmitters whose ranges enclose a same boundary home, the intersection of their communication ranges with the block $t$ defines a set $\cI$ of intervals for which one is the largest.
This largest interval is the communication range of a node $u \in \cS$ that includes all other communication ranges inside of the set $\cI$.
By Fact \ref{f:knownLocation}, all nodes located inside this interval know the limits of the communication range of $u$.
It follows that the potential and range spokesmen for the set of nodes $\cS$ are unique.
These spokesmen are right (left) potential and range spokesmen if the leftmost (rightmost) home of $t$ is completely included in the range of $u$ and not the rightmost (leftmost) home of $t$.

It also follows from the above discussion that for any pair of transmitters $u$ and $v$ whose ranges do not enclose a same boundary home, the spokesmen types defined will be different (right vs. left spokesmen).
\end{proof}

\subsection{Broadcasting Algorithm $\cB$}
In order to complete the broadcasting algorithm, we need a final procedure to transmit the message $m$ from the source to all other nodes of the network.
We now describe Procedure $\cT$.

\begin{algorithm}{}
{\bf Procedure $\cT$}
\begin{algorithmic}
\STATE $S_u \leftarrow $the label of the block,home containing $u$
\STATE {\bf In parallel for all regions}
\FOR{$block = 1..\mu$}
    \IF{$S_u = block$}
        \STATE update spokesman status
        \STATE Spokesmen transmit the message $m$ in the following order:
        \STATE 1) left boundary spokesman,
        \STATE 2) right boundary spokesman,
        \STATE 3) left range spokesman,
        \STATE 4) right range spokesman,
        \STATE 5) left potential spokesman,
        \STATE 6) right potential spokesman
    \ELSE
        \STATE Listen to incoming messages for $6$ rounds
    \ENDIF
\ENDFOR
\end{algorithmic}
\end{algorithm}

\begin{lemma}
\label{l:ProcT}
Procedure $\cT$ broadcasts the message correctly through the network in time $O(D/l)$.
\end{lemma}

\begin{proof}
Consider a network $G$ of diameter $D$ built by the adversary under the swamping model.
Consider also the network $G'$ with the same nodes and links as $G$, but where nodes may receive messages from multiple neighbors in one round without collisions.
Let the broadcasting algorithm $\cF$ execute such that, when a node receives a message $m$ the first time, it transmits this message to all its neighbors the next round.
The  algorithm $\cF$ executes in $\Theta(D)$ rounds on the network $G'$.
We prove the lemma statement by comparing the execution of Procedure $\cT$ on $G$ to the execution of Algorithm $\cF$ on $G'$.

Consider $G$ and the partition $\cP$.
Since each region has $\lceil 3/l \rceil $ blocks, where $l = \max\{\gamma, (1-s)\}$ and since each block has a constant number of spokesmen, the broadcast algorithm sequentially makes all spokesmen of a region communicate every $\Theta(1/l)$ rounds.
From Lemma~\ref{l:noCollide}, the process is collision-free.
From Lemma~\ref{c:rangeOverlap}, the spokesmen of a block reach all the nodes that can be reached by any node on their block that do know the message $m$.
It then follows that the message $m$ being relayed through the network may be slowed down by a factor $O(1/l)$ with respect to the execution of Algorithm $\cF$ in $G'$.
Hence, for any network $G$ of diameter $D$, the total transmission time is in $O(D/l)$.
\end{proof}

\begin{algorithm}{}
{\bf Algorithm $\cB$}
\begin{algorithmic}
\STATE {\bf In parallel} for all nodes
\STATE {\bf Call} Procedure $\cD^*$
\STATE {\bf Call} Procedure $\cT$
\end{algorithmic}
\end{algorithm}

\begin{Proof}{ of Theorem~\ref{t:algoB}}
From Lemma \ref{l:ProcT}, the time of execution of Procedure $\cT$ is $O(D/l)$.
From Lemma \ref{l:procD*}, the time of execution of Procedure $\cD^*$ is $\Theta(g^2)$.
Adding these times together, we get a total time of $O(D/l + g^2)$.
\end{Proof}

\section{Two-Dimensional Lattice}
\label{s:Lattice}
In Section \ref{s:LL}, we have shown an optimal time broadcast algorithm for the lattice line.
We now extend this result to multi-dimensional lattices.
Hence, we consider the set $V$ of $n$ nodes placed at Euclidean coordinates $(i,j)$ for $i=0,1,\ldots,\sqrt{n}-1$ and $j=0,1,\ldots,\sqrt{n}-1$.
We call this placement of the nodes the {\em two-dimensional lattice}.

Consider that each node has a communication range $r$ and a swamping range $s$ such that $r - s > 1$.
Throughout this section, we assume that $r$ and $s$ are positive integers.
We call {\em transmission annulus of $u$} the region at distance greater than $s$ and at most $r$ from a node $u$ and denote it by $A_u$.
Each node $u$ shares a link with each node $v$ located within $A_u$.
The set of links $E$ is the union of all these shared links.
We will present Algorithm $\cA^2$, an extension of Algorithm $\cA$, to broadcast a message in this two-dimensional lattice.
In this section we will prove the following result:

\begin{theorem}
\label{t:A2}
Algorithm $\cA^2$ broadcasts in time $4\lfloor \sqrt{n}/r \rfloor + 12(\lceil r / (r-(s+1)) \rceil + 1)$.
This order of magnitude is optimal.
\end{theorem}

In order to present Algorithm $\cA^2$ and prove the main theorem of this section, we need a preparatory lemma.
Fix one row $l$ of nodes in the square lattice and consider the region covered by all the transmission annuli in the execution of algorithm $\cA$ on this line.

\begin{lemma}
\label{c:thickLine}
Algorithm $\cA$ broadcasts the message to all nodes on $l$ and to all nodes within distance $\lfloor \sqrt{3}r/2 \rfloor$ from $l$.
\end{lemma}

\begin{proof}
In Algorithm $\cA$, nodes broadcast a message along a line.
To do so, nodes at distance at most $r-(s+1)$ from one to the next transmit the message.
The algorithm is successful because each node sends the message symmetrically to intervals of length $r - s$, resulting in complete coverage of the line by the set of transmitting nodes.
It follows that, if each transmitter on a line also covers a length $r-s$ on a parallel lattice line then, complete coverage of this line would be achieved by the end of this Algorithm $\cA$.
For a line $l_d$ parallel to $l$ and at distance $d$ from $l$ and for a fixed node $u$ on the line $l$, let the segment $ls_d$ be a line segment resulting from the intersection of the communication annulus of $u$ and the line $l_d$.

We now evaluate the length of $ls_d$.
Let $L$ denote the length of some line segment $ls_d$.
For $d \leq s$, there are $2$ line segments on each lattice line.
By the law of cosines, we have that the length $L$ is
$$
L^2 = r^2 + s^2 - 2rs\cos(\theta) \geq r^2 - 2rs + s^2 = (r-s)^2.
$$
Hence, for all line segments with one endpoint at distance $s$ from $u$ and another at distance $r$ we have that $L \geq r-s$ for all segments $ls_d$.
See Figure~\ref{fig:range_overlap}.
\begin{figure}[!htb]
\centering
\includegraphics[scale=0.75, bb =0 0 416 96,angle=0]{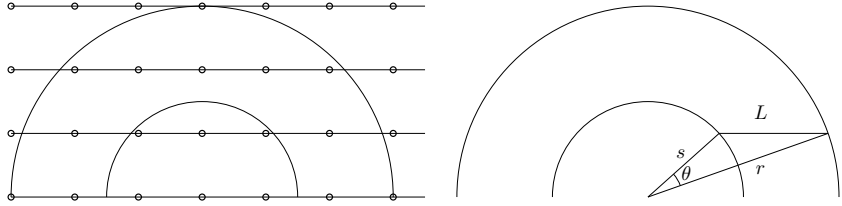}
\caption[Coverage of transmissions]{Coverage of transmissions: the intersection of the communication annulus with lines parallel to $l$ defines the length of the line segments included in the communication annulus.}
\label{fig:range_overlap}
\end{figure}

On the other hand, for $s < d \leq r$, there is a single line segment on each lattice line with length $l_d = 2\sqrt{r^2 - d^2}$.
In this case, for $d \leq \sqrt{3r^2/4 + rs/2-s^2/4}$ we have
$$
l_d \geq 2\sqrt{r^2 - (3r^2/4 + rs/2-s^2/4)} = 2\sqrt{r^2/4 - rs/2 + s^2/4} = 2(r-s)/2.
$$
Since $r>s$, this is also true if $d \leq \sqrt{3}r/2$.

Hence, Algorithm $\cA$ completes message dissemination on all lines within distance $\lfloor \sqrt{3}r/2 \rfloor$ from $l$.
\end{proof}

Building on Algorithm $\cA$ and on the fact that it broadcasts the message $m$ to all nodes on a line $l$ and to all nodes within distance $d = \lfloor\sqrt{3}r/2\rfloor$ from $l$,
Algorithm $\cA^2$ operates in $4$ phases:
\begin{enumerate}
\item Execution of Algorithm $\cA$ on the horizontal line $l$ of the source node.
\item Execution of Algorithm $\cA$ on the vertical lines at coordinates $d, 7d, 13d, \ldots$
\item Execution of Algorithm $\cA$ on the vertical lines at coordinates $3d, 9d, 15d, \ldots$
\item Execution of Algorithm $\cA$ on the vertical lines at coordinates $5d, 11d, 17d, \ldots$
\end{enumerate}

As it was the case for Algorithm $\cA$, algorithm $\cA^2$ is valid for any source node.
We now prove the main theorem of this section.

\begin{Proof}{ of Theorem~\ref{t:A2}}
By Lemma~\ref{c:thickLine}, and considering the set of lines at distance $2d = 2\lfloor\sqrt{3}r/2\rfloor$, Algorithm $\cA^2$ will achieve complete coverage of the lattice within its execution, given that no collision occurs.

We now demonstrate that no collision occurs from simultaneous transmissions.
To show this, we show that any two nodes transmitting in one round are at distance at least $2r$ from one-another.
Any simultaneous transmission occurs from nodes at distance $6 \lfloor\sqrt{3}r/2\rfloor$ from one-another.
We have that the distance between any two nodes transmitting in the same round is
\begin{eqnarray*}
6 \lfloor\sqrt{3}r/2\rfloor &>& 3\sqrt{3}r - 6 > 5.19r - 6\\
    &=& 2r + (3.19r - 6) > 2r \mbox{ for $r \geq 2$.}
\end{eqnarray*}
An implication of Lemma~\ref{l:impossible} is that broadcasting in lattice networks with swamping is impossible for $r < 2$.
Hence the assumption that $r \geq 2$ is true in all cases when broadcast is possible.

Since Algorithm $\cA^2$ is a sequence of $4$ executions of Algorithm $\cA$, it runs in the time of $4$ execution of algorithm $\cA$ on lines of length $\sqrt{n}$.
Hence, the execution time of Algorithm $\cA^2$ is $4\lfloor \sqrt{n}/r \rfloor + 12(\lceil r / (r-(s+1)) \rceil + 1)$.

We now prove that our algorithm is of optimal time complexity.
Nodes within the swamping radius of a transmitting node cannot receive any message;
the maximum number of nodes which can receive a message in one round within the communication radius of a node is then the nodes within its annulus.
Consider a set of nodes $N$, at distance $L$ from the source node, where $L$ is some multiple of $r$.
Consider further that this set $N$ is the intersection of a disk of diameter $r$ and the square lattice.
At least $L/r$ rounds are needed for the message to reach the set $N$.
At the following round, broadcasting within $N$ may begin.
From the communication model, at most $\pi(r^2 - s^2)$ nodes inside $N$ may receive the message within each round.
Hence, because $2\pi r(r-s) > \pi(r^2-s^2)$, less than $2 \pi r (r - s)$ nodes receive the message in each round.
Therefore, the total broadcasting time is at least
$$
\left\lceil \frac{L}{r} \right\rceil + \left\lceil \frac{\pi r^2}{2\pi r(r - s)} \right\rceil = \left\lceil \frac{L}{r} \right\rceil + \left\lceil \frac{r}{2(r - s)} \right\rceil.
$$
For communication in the $n$-node square lattice, we obtain a lower bound on the broadcasting time which is
$$
\left\lceil \frac{\sqrt{n}}{2r} \right\rceil + \left\lceil \frac{r}{2(r - s)} \right\rceil \in \Omega\left( \frac{\sqrt{n}}{2r} + \frac{r}{2(r - s)}\right),
$$
for $r > 1$ and $s$ integers and for $r - s > 1$.
This lower bound matches the time complexity of Algorithm $\cA^2$.
\end{Proof}

By the same technique used to extend Algorithm $\cA$ to Algorithm $\cA^2$, we may extend the Algorithm $\cA$ to an algorithm $\cA^d$, broadcasting in the $d$-dimensional lattice, when $d\in \Theta(1)$.
By a proof similar to that of Theorem~\ref{t:A2}, we may show the following result.

\begin{lemma}
Algorithm $\cA^d$ broadcasts in the $d$-dimensional lattice, $d\in \Theta(1)$, with time complexity in $\Theta(d \sqrt[d]{n} + d r/(r-(s+1)))$.
\end{lemma}

\section{City Model}
\label{s:City}
We now consider the task of broadcasting in a connected network of unknown topology.
In particular, we consider networks with nodes placed at points on the plane, located at least at some geometric distance $\gamma$ from each-other.
Each node $u$ is equipped to communicate with all nodes that are both within distance $1$ and at distance greater than $s$ from it.
Hence, in this section, we assume that $r=1$ for simplicity.
More formally, we describe the {\em city model}.
The communication range of a node $u$ is the annulus centered at $u$ with radii $s$ and $1$.
The {\em size} of the communication range is the width of this annulus,
i.e., $1-s$.
The adversary designs the network such that it is {\em connected} and the distance between any pair of nodes $u,v$ is at least $\gamma$.
We say that a network is connected if, for any node pair $u,v$, there exists a path in the network from node $u$ to node $v$.
Observe that the network is connected only if $\gamma \leq 1$.

Nodes are aware of the parameter $\gamma$ (and $g=1/\gamma$) and the coordinate system of the plane.
Each node also knows the parameter $s$, its swamping distance, and its communication distance $1$.

We wish to complete broadcasting in a collision avoidance scheme.
We will use the assumption of spontaneous wake-up of the nodes.
In this section we use the apparent silence from collisions to discover the presence of nodes through a collision-causing process, used before the transmission part of the broadcasting algorithm.
Moreover, we will assume that nodes know about the transmissions made within close proximity.
We will show the following result.

\begin{theorem}
\label{t:algoB2}
Algorithm $\cB^2$ broadcasts a message $m$ in a network of diameter $D$ in time $O(Dg/l + g^4)$, where $l = \max\{(1-s)/(3\sqrt{2}),\gamma/\sqrt{2}\}$.
\end{theorem}

\subsection{Partition $\cP^2$ of the Plane}
We now define a partition, called $\cP^2$, on which our communication algorithm will operate.

Each square in the partition below includes its North border, its West border, and both its North vertices;
it excludes its East border, its South border and both its South vertices.
We provide a graphical representation of the partition in Figure \ref{fig:partition2D} and now describe it below.
\begin{figure}[!htb]
\centering
\includegraphics[scale=0.5, bb =0 0 552 192,angle=0]{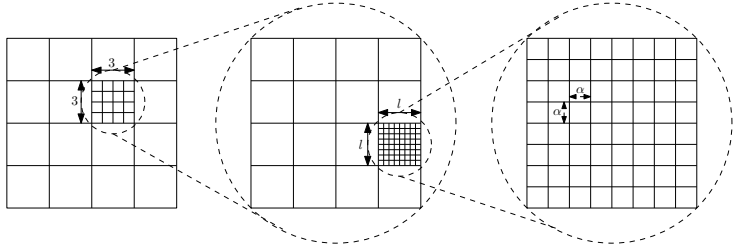}
\caption[Partition $\cP^2$]{Partition $\cP^2$: from left to right, the plane is partitioned into $3\times 3$ squares called regions;
for $l = \max\{(1-s)/(3\sqrt{2}),\gamma/\sqrt{2}\}$ regions are partitioned into $l \times l$ squares called blocks;
blocks are partitioned into $\gamma/\sqrt{2} \times \gamma/\sqrt{2}$ squares called homes.}
\label{fig:partition2D}
\end{figure}

Partition the plane into a mesh of $3 \times 3$ squares called {\em regions}.

Further partition each region into a mesh of $l \times l$ squares, called {\em blocks}, with length $l = \max\{(1-s)/(3\sqrt{2}),\gamma/\sqrt{2}\}$.
Here, $l \leq 1/\sqrt{2}$ since both $\gamma \leq 1$ and $(1-s)/3 \leq 1$.
Each region contains $\mu = \lceil 3/l \rceil^2$ blocks, where $\lfloor 3/l \rfloor^2$ are of area $l^2$ and at most $2 \lfloor 3/l \rfloor + 1$ are smaller, and even may consist of a single line or point.
For each region, label blocks $1, 2, \ldots, \mu$, from West-East row by row, North to South.

Partition also each block into a mesh of $\gamma/\sqrt{2} \times \gamma/\sqrt{2}$ squares, called {\em homes}.
Each block contains $\nu = \lceil \sqrt{2}l/\gamma \rceil^2$ homes, where $\lfloor \sqrt{2}l/\gamma \rfloor^2$ are of area $\gamma^2/2$ and at most $2\lfloor \sqrt{2}l/\gamma \rfloor + 1$ are smaller, and even may consist of a single line or point.
For each block, label homes sequentially $1, 2, \ldots, \nu$, from West-East row by row, North to South.

\subsubsection{Partition Properties}
In section~\ref{s:Highway}, we showed properties for the partition $\cP$.
We now show the validity of lemmas \ref{l:singleHome} and \ref{l:noCollide} for the partition $\cP^2$.
For ease of reading, we now repeat these lemmas.

\begin{relemma}{\ref{l:singleHome}}
Each home contains at most one node.
\end{relemma}

Observe that since homes have diameter at most $\gamma$, at most one node can occupy each home.
Hence, Lemma~\ref{l:singleHome} holds for partition $\cP^2$.

\begin{relemma}{\ref{l:noCollide}}
Transmissions from unique nodes inside identically labeled blocks in distinct regions do not collide.
\end{relemma}

Consider nodes $u,v$ in different regions and identically labeled blocks.
Since each region has side length $3$ and each block has side length $l \leq 1$, Lemma~\ref{l:noCollide} holds for $\cP^2$.

In the following sections, we describe communication procedures that will enable nodes to broadcast messages to all nodes of their networks.

\subsection{Procedure $\cD$ for Nodes in Range}
Recall Procedure $\cD$ in which nodes send a message sequentially based to their $(block,\,home)$ label.
\begin{algorithm}{}
{\bf Procedure $\cD$}
\begin{algorithmic}
\STATE {\bf In parallel for all nodes $u \in V$}
\STATE $N_u \leftarrow \emptyset$ // the set of nodes known to $u$
\STATE $H_u \leftarrow$ the (block, home) label of $u$
\FOR{$block = 1..\mu$}
    \FOR{$home = 1..\nu$}
	   \IF{$H_u = (block, home)$}
            \STATE Transmit $hello$
        \ELSIF{a $hello$ is  heard}
            \STATE $N_u \leftarrow N_u \cup (block, home)$
        \ENDIF
    \ENDFOR
\ENDFOR
\end{algorithmic}
\end{algorithm}
Since Lemma~\ref{l:singleHome} and Lemma~\ref{l:noCollide} both hold for $\cP^2$, we have that Lemma~\ref{l:procD} also still holds for $\cP^2$.

\begin{relemma}{\ref{l:procD}}
Upon completion of Procedure $\cD$, nodes know of all nodes within distance $1$ and greater than $s$ of them.
\end{relemma}

By repeating Procedure $\cD$ $i$ times (augmenting the $hello$ message with the location of known nodes), a node $u$ can learn about other nodes within hop distance $i$ ($\Gamma_{\leq i}(u)$).
However, the hop distance from a node $u$ to a node $v$ may be arbitrarily large, even if $v$ is within geometric distance $1$ of $u$.

Hence, for diameter $D$ graphs, the use of Procedure $\cD$ alone could take as many as $D g^2$ rounds to discover the existence of all nodes within distance $1$.
In this case, the message could be transmitted without the assumption of spontaneous wake up from the source to the nodes.
We have the following lemma:
\begin{lemma}
In all networks where nodes are placed on the plane, of diameter $D$ and granularity $g$, the broadcast time is in $O(D g^2)$.
\end{lemma}

Hence, under our communication model, Procedure $\cD$ is insufficient to speed up broadcast in the spontaneous wake-up model as opposed to the conditional wake-up model, in the worst case.

\subsection{Procedure $\cD^*$ for Neighborhood Discovery}
Recall Procedure $\cD_{(b, h)}$ using collisions to discover nodes within distance $s$.
\begin{algorithm}{}
{\bf Procedure $\cD_{(b, h)}$}
\begin{algorithmic}
\STATE // $N_u$ is the set of nodes known to $u$
\STATE // $H_u$ is the (block, home) label of $u$
\STATE {\bf In parallel for all nodes $u \in V$}
\FOR{$block = 1..\mu$}
    \FOR{$home = 1..\nu$}
    	\IF{$H_u = (block, home)$ {\bf OR} $H_u = (b, h)$}
	       \STATE Transmit $hello$
	   \ELSIF{no $hello$ is heard {\bf AND} $(b, h) \in N_u$}
	       \STATE $N_u \leftarrow N_u \cup (block, home)$
	   \ENDIF
    \ENDFOR
\ENDFOR
\end{algorithmic}
\end{algorithm}

\begin{lemma}
\label{l:procDbh2D}
By Procedure $\cD_{(b,h)}$, nodes neighbor to $(b,h)$ know all other nodes within geometric distance $1$ of them in time $\Theta(g^2)$.
\end{lemma}
\begin{proof}
The time complexity of Procedure $\cD_{(b,h)}$ is in $\Theta(\mu\nu)$.
With $\gamma \leq l \leq 1$, we have that $\mu = \lceil 3/l \rceil^2 \in \Theta((1/l)^2)$ and $\nu = \lceil \sqrt{2}l/\gamma \rceil \in \Theta((l/\gamma)^2)$.
Hence,
$$
\mu\nu \in \Theta((1/l)^2(l/\gamma)^2) = \Theta((1/\gamma)^2) = \Theta(g^2).
$$

We now prove correctness.
Consider the execution of Procedure $\cD_{(b,h)}$, during which the node $(b,h)$ will transmit messages at every round.
A message from $(b,h)$ will be heard by $u$ at every round when no collision occurs at $u$.
Furthermore, when no message can be distinguished, another node within distance $1$ of $u$ must be transmitting from the home with label $(block,home)$ (as defined in the procedure).
Since Procedure $\cD_{(b,h)}$ schedules all nodes to transmit in pairs with $(b,h)$,
upon completion of this procedure, the node $u$ will have discovered all nodes $w$ for which the geometric distance $dist(u,w)$ from $u$ is at most $1$.
\end{proof}

Recall Procedure $\cD^*$ consisting of one execution of Procedure $\cD$ followed by the execution of Procedure $\cD_{(b,h)}$ for all $(b,h) \in \{1, 2, \ldots, \mu\} \times \{1, 2, \ldots, \nu)\}$.
For the plane, Procedure $\cD^*$ allows the discovery of nodes within distance $1$.
More formally, refer to the pseudo code for Procedure $\cD^{*}$.
\begin{algorithm}{}
{\bf Procedure $\cD^{*}$}
\begin{algorithmic}
\STATE {\bf Call} Procedure $\cD$
\FOR{$b = 1..\mu$}
    \FOR{$h = 1..\nu$}
    	\STATE {\bf Call} Procedure $\cD_{(b, h)}$
    \ENDFOR
\ENDFOR
\end{algorithmic}
\end{algorithm}

Procedure $\cD^*$ accomplishes the same function in the plane as it does in the line however, with increased time complexity.

\begin{lemma}
\label{l:procD*2D}
By Procedure $\cD^*$, nodes know all other nodes within distance $1$ of them in time $\Theta(g^4)$.
\end{lemma}
\begin{proof}
The time complexity of Procedure $\cD_{(b,h)}$ is in $\Theta(\mu\nu)$.
Hence, the time complexity of Procedure $\cD^*$ is in $\Theta(\mu^2\nu^2)$.
By the above and by Lemma~\ref{l:procDbh2D}, the time complexity of Procedure $\cD^*$ is therefore in $\Theta(g^4)$.

We now prove correctness.
For any node $u$, since the graph is connected, by Fact~\ref{f:connect} there exists a node $(b,h)$ such that Procedure $\cD_{(b,h)}$ will be executed.
By the above and by Lemma~\ref{l:procDbh2D}, all nodes know all other nodes that are within distance $1$ of them.
\end{proof}

It remains open whether or not the time $\theta(g^4)$ for neighbourhood discovery is optimal.

One degenerate case of swamping is when $s < \gamma$;
then, swamping has no real effect on the network, which becomes identical to a congruent GRN.
In that case, the process of discovering the neighbours takes only the time necessary for all nodes to announce their presence once, $\theta(g^2)$.
This is only true because of the absence of nodes with which nearby nodes can not communicate.

However, once we have $s \geq \gamma$, some links of the congruent GRN are deleted in the network with swamping;
to discover nodes at close proximity then becomes a non-trivial, collaborative task.
When messages must remain small, it is impossible to share locations of many other nodes to speed up the neighbourhood discovery process.

For small enough $s$ and unbounded message size, the task of neighbourhood discovery may be sped up, but only under special conditions.
In the 2-dimensional case, the length of paths between nearby nodes seems bounded only by the diameter $D$.
Therefore, if nodes transmit long messages containing their current known mapping of neighbour nodes on a turn basis, repeating this process $D$ times to allow dissemination of these maps, i.e., in time $\in \Theta(Dg)$.
The value of $D$ may be much larger than $g^4$.

With knowledge of all nodes within distance $1$, nodes have the basic tools to select distinguished nodes to relay messages for all nodes of a block.
We discuss such a procedure in the following subsection.

\subsection{Selection of Spokesman Nodes}
In this section, we assume that nodes know which nodes of their own block possess the source message $m$.
The spokesmen nodes are those nodes in each row, column and diagonal of homes within a block which possess the message and which are located in the home which is closest to either end of that row, column or diagonal.
We now state the following lemma.
\begin{lemma}
\label{l:goodSpokesman}
If all spokesmen of a block $b$ transmit in a collision-avoidance scheme, then all nodes neighbor to any node in $b$ will receive the source message.
\end{lemma}

The proof will be given following some preliminary facts and discussion.
More formally, the rules for deciding which nodes are spokesmen are as follows:
For a row (column) of homes of partition $\cP^2$, among nodes possessing the message, those two nodes in homes closest to the West and East (North and South) borders of a block in $\cP^2$ are spokesmen.
For a diagonal of homes of partition $\cP^2$, among nodes possessing the message, those two nodes in homes closest to the borders of a block in $\cP^2$ are spokesmen.
See Figure~\ref{fig:dynspokesmen}.
\begin{figure}[!htb]
\centering
\includegraphics[scale=0.75, bb = 0 0 224 224,angle=0]{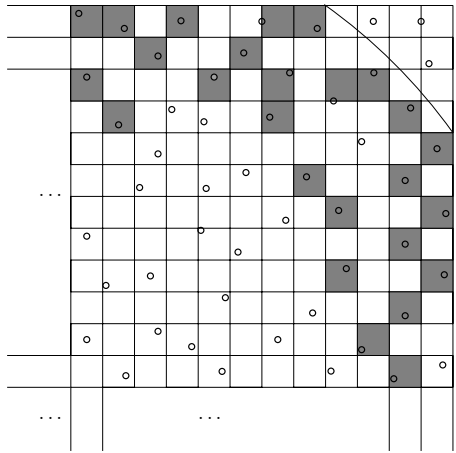}
\caption{The spokesmen of a block}
\label{fig:dynspokesmen}
\end{figure}

If a spokesman is chosen in column (row) $i$ because of its proximity to the North or South (West or East) border, then it has the label $N_i$ and/or $S_i$, resp. ($W_i$ and/or $E_i$, resp.).
If a spokesman is chosen in Southeast-Northwest (Southwest-Northeast) diagonal $i$ because of its proximity to the Southeastern or Northwestern (Southwestern or Northeastern) border, then it has the label $SE_i$ and/or $NW_i$, resp. ($SW_i$ and/or $NE_i$, resp.).
Spokesmen can be assigned more than one such label.

Observe that there are $O(l^2g^2)$ homes inside a block;
there are $O(lg)$ rows of homes, $O(lg)$ columns of homes and $O(lg)$ diagonals of homes inside a block;
there are at most $2$ spokesmen elected for each row, each column and each diagonal.
Hence, each block contains $O(lg)$ spokesmen.
We now claim that only these spokesmen are necessary to broadcast.

Before presenting the proof, we recall the following fact.
\begin{fact}
\label{f:halfplane}
Consider two vertices $A$ and $B$ and the line $\overline{AB}$ joining them.
The line $l$ perpendicular to $\overline{AB}$ and through its center defines two halfplanes $H_{A,B}$ and $H_{B,A}$.
The halfplane $H_{A,B}$ (resp. $H_{B,A}$) contains $A$ ($B$) and has all points closer to $A$ ($B$) than to $B$ ($A$).
\end{fact}

We now proceed to the presentation of two preparatory lemmas: Lemma~\ref{l:closer} and Lemma~\ref{l:farther}.
Using these lemmas, we will then prove Lemma~\ref{l:goodSpokesman}.
\begin{lemma}
\label{l:closer}
The set of spokesmen of a block is closer to any point $p$ outside the block than any non-spokesman node.
\end{lemma}
\begin{proof}
Consider the sector $S$ of a plane defined by the angle $ACB$ of a triangle.
We first show that if the angle $\theta$ at $C$ is at most $\pi/2$, then all points in the sector outside the triangle $ACB$ are closer to $A$ and $B$ than they are to $C$.

Consider the halfplanes defined by the vertex pairs $A,C$ and $B,C$ as described in Fact~\ref{f:halfplane}.
If the node $C$ is closer than $A$ and $B$ to a point $p$, then $p$ is in the intersection of $H_{C,A}$ and $H_{C,B}$.
Moreover, if $\theta = \pi/2$, then $S \cap H_{C,A} \cap H_{C,B}$ is a rectangle contained within the triangle $ACB$.
As $\theta$ decreases, the region $S \cap H_{C,A} \cap H_{C,B}$ remains contained within the triangle $ACB$.
See Figure~\ref{fig:triangle_domination}.
\begin{figure}[!htb]
\centering
\includegraphics[scale=0.75, bb = 0 0 304 80,angle=0]{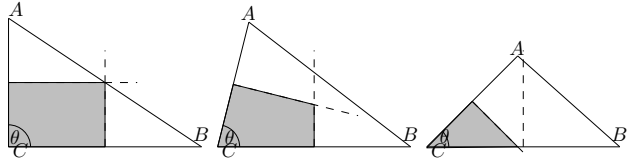}
\caption[Proximity by sector of spokesmen]{Proximity by sector of spokesmen: $A$ and $B$ are in the halfplanes containing all points of the sector not in the triangle $ACB$.}
\label{fig:triangle_domination}
\end{figure}
It follows that all other points of $S$ are closer to either $A$ or $B$.

Now consider a non-spokesman node $u$ and the set of all spokesmen in its row, column and diagonals.
For $u$ not to be a spokesman, it must have one spokesman on each side of itself for its row, column and diagonals.
Let these spokesmen be labeled sequentially $u_1, u_2, \ldots, u_8$ in a clockwise order around the node $u$.
Consider a partitioning of the plane around $u$ by the set of half-lines starting at $u$ and
going through $u_1, u_2, \ldots, u_8$.
Call these plane regions the sectors $u_i\,u\,u_{i+1}$.

Since the distance between nodes is at least $\gamma$ and because of the geometry of the partition, we have that the angle of each sector $u_i\,u\,u_{(i+1)\mod 8}$ is less than $\pi/2$.
By the first part of the argument, the node $u$ is farther from any point in a sector $u_i\,u\,u_{(i+1)\mod 8}$, and outside the triangle $u_i\,u\,u_{(i+1)\mod 8}$, than the spokesmen $u_i$ and $u_{(i+1)\mod 8}$.
See Figure~\ref{fig:spokesman_domination}.
\begin{figure}[!htb]
\centering
\includegraphics[scale=0.75, bb = 0 0 112 108,angle=0]{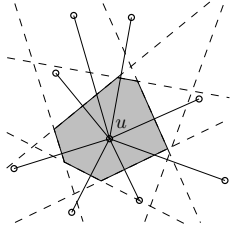}
\caption[Overall proximity of spokesmen]{Overall proximity of spokesmen: for all points outside of the gray region, there is always a spokesman that is closer than $u$.
For all points, there is always a spokesman that is farther than $u$.}
\label{fig:spokesman_domination}
\end{figure}
\end{proof}

\begin{lemma}
\label{l:farther}
For any point $p$ and any non-spokesman node $u$, there is always a spokesman node $v$ that is farther from $p$ than $u$.
\end{lemma}
\begin{proof}
Consider a non-spokesman node $u$ and the set of all spokesmen in its row, column and diagonals.
For $u$ not to be a spokesman, it must have one spokesman on each side of itself for its row, column and diagonals.
Let these spokesmen be labeled sequentially $u_1, u_2, \ldots, u_8$ in a clockwise order around the node $u$.
Recall Fact~\ref{f:halfplane}.
Consider all halfplanes $H_{u,u_i}$, $i=1,2,\ldots,8$.
These halfplanes contain all those points to which $u$ is closer than $u_i$, or those from which $u_i$ is farther than $u$.
Since the distance between nodes is at least $\gamma$ and because of the geometry of the partition, we have that each angle $u_i\,u\,u_{(i+1)\mod 8}$ is less than $\pi/2$.
Therefore, the union of these halfplanes covers the entire plane.
See Figure~\ref{fig:spokesman_domination}.
\end{proof}

We now prove the main lemma of this section.

\begin{Proof}{ of Lemma~\ref{l:goodSpokesman}}
Fix a node $u$ inside the block $b$.
Fix a node $v$ neighbor of $u$.
If $u$ is a spokesman, we are done.
Otherwise, we must show that there is a spokesman $w$ that shares a link with $v$.

If $u$ is not a spokesman, then from Lemma~\ref{l:closer} and from Lemma~\ref{l:farther}, there is a spokesman $w$ that is closer to $v$ than $u$ and there is a spokesman $w'$ that is farther.
For some $\delta$, $w$ is at distance $dist(u,v) - \delta < dist(v,w) < dist(u,v)$ of $v$
and $w'$ is at distance $dist(u,v) < dist(v,w') < dist(u,v) + \delta$ from $v$.
Since $u$ shares a link with $v$, we know that $s < dist(u,v) < 1$.
Moreover, for $\delta < (1-s)/2$, either $s < dist(u,v) - \delta < 1$ or $s < dist(u,v) + \delta < 1$.
Since the diameter of a block is $(1-s)/3 < (1-s)/2$, at least one of $w$ and $w'$ shares a link with $v$.
\end{Proof}

\subsection{Broadcasting Algorithm $\cB^2$}
In order to complete the broadcasting algorithm, we need a final procedure to transmit the message $m$ from the source to all other nodes of the network.
We now describe Procedure $\cT^2$.
Procedure $\cT^2$ is executed in parallel for all regions.
Sequentially for all blocks, we have the set of spokesmen transmit the message $m$ on a turn basis.
Spokesmen send the message only once each and the procedure ends implicitly when the last message is sent.
More formally, refer to the pseudo code for Procedure $\cT^2$.

\begin{algorithm}{}
{\bf Procedure $\cT^2$}
\begin{algorithmic}
\STATE $S_u \leftarrow $the label of the $(block, home)$ containing $u$
\STATE {\bf In parallel for all regions}
\REPEAT
    \FOR{$block = 1..\mu$}
        \STATE update spokesman status
        \IF{$u$ is a spokesman AND $S_u = block$ AND $u$ has not sent the message}
            \FOR{valid row indices $i = 1,\ldots$}
                \STATE Spokesmen $E_i$, $W_i$ transmit the message $m$ in order
            \ENDFOR
            \FOR{valid column indices $i = 1,\ldots$}
                \STATE Spokesmen $N_i$, $S_i$ transmit the message $m$ in order
            \ENDFOR
            \FOR{valid diagonal indices $i = 1,\ldots$}
                \STATE Spokesmen $NE_i$, $SE_i$, $NW_i$, $SW_i$ transmit the message $m$ in order
            \ENDFOR
        \ELSE
            \STATE Listen to incoming messages until all spokesmen have transmitted
        \ENDIF
    \ENDFOR
\UNTIL{no node has transmitted in an iteration}
\end{algorithmic}
\end{algorithm}

\begin{lemma}
\label{l:ProcT2}
Procedure $\cT^2$ broadcasts the message correctly through the network in time $O(Dg/l)$.
\end{lemma}

\begin{proof}
Consider a network $G$ of diameter $D$ built by the adversary under the swamping model.
Consider also the network $G'$ with the same nodes and links as $G$, but where nodes may receive messages from multiple neighbors in one round without collision.
Let the nodes of $G'$ execute the broadcasting algorithm $\cF$: when a node receives a message $m$ the first time, it transmits this message to all its neighbors the next round.
For the network $G'$, the algorithm $\cF$ executes in $\Theta(D)$ rounds.
We prove the lemma statement by comparing the execution of Procedure $\cT^2$ on $G$ to the execution of Algorithm $\cF$ on $G'$.

Since each region has $\lceil 3/l \rceil^2$ blocks, where $l = \max\{\gamma/\sqrt{2}, (1-s)/(3\sqrt{2})\}$ and since each block has $O(lg)$ spokesmen, the broadcast algorithm sequentially makes all spokesmen of a region communicate every $O(g/l)$ round.
From Lemma~\ref{l:noCollide}, the process is collision-free.
From Lemma~\ref{l:goodSpokesman}, the spokesmen of a block reach all the nodes that can be reached by any node on their block that do know the message $m$.
It then follows that the message $m$ being relayed through the network may be slowed down by a factor $O(g/l)$ with respect to the broadcast time of Algorithm $\cF$.
Hence, for any network $G$ of diameter $D$, the total transmission time is in $O(Dg/l)$.
\end{proof}

\begin{algorithm}{}
{\bf Algorithm $\cB^2$}
\begin{algorithmic}
\STATE {\bf In parallel} for all nodes
\STATE {\bf Call} Procedure $\cD^*$
\STATE {\bf Call} Procedure $\cT^2$
\end{algorithmic}
\end{algorithm}

\begin{Proof}{ of Theorem~\ref{t:algoB2}}
From Lemma \ref{l:ProcT2}, the time of execution of Procedure $\cT^2$ is $O(Dg/l)$.
From Lemma \ref{l:procD*}, the time of execution of Procedure $\cD^*$ is $\Theta(g^4)$.
Adding these times together, we get a total time of $O(Dg/l + g^4)$.
\end{Proof}

\section{Conclusion}
In this paper, we have shown algorithms for broadcasting under a novel communication model, the swamping communication model.
We have shown algorithms of optimal time complexity for the line and the grid.
We have also shown algorithms for broadcasting in networks of unknown topology, with nodes placed on the line, and in the plane.

In \cite{EGKPPS2009}, under the spontaneous wake up model\index{network!spontaneous wake up model}, where nodes may transmit from the beginning of the communication process, the authors combined two sub-optimal algorithms into one algorithm, which completes broadcasting in optimal time $O(\min(D + g^2, D \log g)$.
Comparatively, our algorithm is slower by a factor of, at least, $g/l$.
The reason for this slowdown is the presence of swamping and the complexity of the ensuing collision-avoidance broadcasting scheme.
Contrary to the cited work, it is not possible to select one node per region to transmit the message to all nearby nodes.
The lower bound on the time complexity for broadcasting in the presence of swamping remains open.

\section*{Acknowledgements}
Thanks go to Dr. Ioannis Lambadaris for suggesting the swamping paradigm and for useful conversations on this topic.
Thanks also go to Dr. Andrzej Pelc for useful conversations on the swamping communication model.

\bibliographystyle{plain}
\bibliography{biblio}

\end{document}